\documentclass[12pt]{article}
\usepackage[utf8]{inputenc}
\usepackage[american]{babel}
\usepackage{amsfonts, amsmath, amssymb}
\usepackage{mathtools}
\usepackage{fancyhdr}
\usepackage{amsthm}
\usepackage{graphicx}
\usepackage{changepage}
\usepackage[
backend=biber,
style=alphabetic,
]{biblatex}
\newtheorem{thm}{Theorem}

\theoremstyle{definition}

\title{A remark on conditional entropy}
\author{Adam Wang}
\date{Mar 2024}

\begin{document}

\maketitle

\abstract{The following note proves that conditional entropy of a sequence is almost time-reversal invariant, specifically they only differ by a small constant factor dependent only upon the forward and backward models that the entropies are being calculated with respect to. This gives rise to a numerical value that quantifies learnability, as well as a methodology to control for distributional shift between datasets. Rough guidelines are given for practitioners.}

\section{Introduction}
Entropy and its variants are widely utilized metrics within information theory and machine learning. Many theorems and applications have been found, with most of them stemming from the astounding fact that this quantity is minimized exactly when two distributions are equivalent \cite{mackay_information_2019}.

Conditional entropy is one variant that has found usage in many different fields \cite{porta_conditional_1999}\cite{graham_using_2013}, such as the theoretical underpinnings of context modelling in file compression, though it is just referred to as entropy in this context.  

The following manuscript notes that conditional entropy of sequential datasets is almost invariant to forward and backward passes, akin to a time-reversal invariance. This gives rise to a way of quantifying learnability. The manuscript will focus on the discrete case. 

If one chooses to model a sequential dataset $S$ from an alphabet of finitely many symbols $\mathcal X$ with a predictive model of a fixed length context of $n$, one is trying to find a conditional probability distribution, which in proxy assumes such a distribution exists:
\begin{align*}
    p(x_{n}=s_n|x_0=s_0,\cdots,x_{n-1}=s_{n-1})
\end{align*}
Thus, we assume that our dataset is generated by some process that follows this fixed probability distribution. If one further makes the assumption that for our process there is a well defined probability distribution for sequences of length $n$:
\begin{align*}
    p(x_0=s_0,\cdots,x_{n-1}=s_{n-1})=p(s_0,\cdots,s_{n-1})
\end{align*}
Then, we can always find the "reverse" conditional probabilities to predict the previous symbol given the latter:
\begin{align*}
    p(x_0=s_0|x_1=s_1,&\cdots,x_{n}=s_{n}) = \\
    &\frac{p(x_{n}=s_n|x_0=s_0,\cdots,x_{n-1}=s_{n-1})p(s_0,\cdots,s_{n-1})}{p(s_1,\cdots,s_{n})}
\end{align*}
Thus given a sequence of symbols $S$ that follows the above process, we can define its reverse to be $\hat S$. Thus $\hat S$ must follow the distribution $\hat p$ given by:
\begin{align*}
    \hat p(x_{n}=s_n|x_0=s_0,\cdots,x_{n-1}=s_{n-1}) = p(x_{0}=s_n|x_{1}=s_{n-1},\cdots,x_{n}=s_{0})
\end{align*}
Henceforth, unambiguously we will write conditional $p$ to be predicting the next symbol given the previous $n$ in $S$ and $\hat p$ to be predicting the previous symbol given the next $n$ in $S$. These notions are reversed in $\hat S$. We also have that:
\begin{align*}
    p(s_0,\cdots,s_{n-1}) = \hat p(s_{n-1},\cdots,s_0)
\end{align*}
Next, we define the following symbol:
\begin{align*}
    \#_{S,n+1}(s_0\cdots s_n)
\end{align*}
to mean the number of occurrences of the $n+1$-tuple within the dataset $S$. If the tuple is shorter, for example:
\begin{align*}
    \#_{S,n+1}(s_0\cdots s_{n-1})
\end{align*}
This will denote the number of occurrences of $n+1$-tuples that start with $n$-tuple within the dataset.

Armed with these gadgets, we are now ready to prove the main theorem. 

\begin{thm}
    For a sequential dataset $S$ of length $N$ generated by some process with well defined conditional and unconditional distribution, $p$, the difference between forward and backward conditional entropy are given by:
    \begin{align*}
        H_p(S)-H_{\hat p}(\hat S)&=\log(p(\vec x_f))-\log(p(\vec x_l))\\
        &\leq C
    \end{align*}
    Where $\vec x_f, \vec x_l$ are the first and last $n$-tuples of $S$ and $C$ is a constant dependent only upon $p$. In other words, the difference in average conditional entropy is $\mathcal O(1/N)$
\end{thm}
\begin{proof}
    Consider:
    \begin{align*}
        H_p(S) &= -\sum_{(x_0\cdots x_n)\in S} \log(p(x_n|x_0,\cdots,x_{n-1})) \\
        &= -\sum_{(s_0\cdots s_n)\in \mathcal X^{n+1}} \#_{S,n+1}(s_0\cdots s_n) \log(p(s_n|s_0,\cdots,s_{n-1})) \\
        &= -\sum_{(s_0\cdots s_n)\in \mathcal X^{n+1}} \#_{S,n+1}(s_0\cdots s_n) (\log(p(s_0,\cdots,s_n)) \\
        &\hspace{15em}- \log(p(s_0,\cdots,s_{n-1})))\\
        &= H_1 + H_2
    \end{align*}
    One can repeat this step for the reverse entropy:
    \begin{align*}
        H_{\hat p}(\hat S) &= -\sum_{(x_0\cdots x_n)\in \hat S} \log(\hat p(x_n|x_0,\cdots,x_{n-1})) \\
        &= -\sum_{(s_0\cdots s_n)\in \mathcal X^{n+1}} \#_{\hat S,n+1}(s_0\cdots s_n) (\log(\hat p(s_0,\cdots,s_n)) \\
        &\hspace{15em} - \log(\hat p(s_0,\cdots,s_{n-1})))\\
        &= H'_1 + H'_2
    \end{align*}
    Notice that by relabelling we get that:
    \begin{align*}
        H'_1&=-\sum_{(s_0\cdots s_n)\in \mathcal X^{n+1}} \#_{\hat S,n+1}(s_0\cdots s_n) \log(\hat p(s_0,\cdots,s_n))\\
        &= -\sum_{(s_0\cdots s_n)\in \mathcal X^{n+1}} \#_{S,n+1}(s_n\cdots s_0) \log(p(s_n,\cdots,s_0))\\
        &= H_1
    \end{align*}
    Now let us consider $H_2$:
    \begin{align*}
       H_2 &= \sum_{(s_0\cdots s_n)\in \mathcal X^{n+1}} \#_{S,n+1}(s_0\cdots s_n) \log(p(s_0,\cdots,s_{n-1})) \\
        &=\sum_{(s_0\cdots s_{n-1})\in \mathcal X^{n}} \log(p(s_0,\cdots,s_{n-1})) \sum_{s_n\in \mathcal X}\#_{S,n+1}(s_0\cdots s_n)\\
        &=\sum_{(s_0\cdots s_{n-1})\in \mathcal X^{n}} \#_{S,n+1}(s_0\cdots s_{n-1})\log(p(s_0,\cdots,s_{n-1}))
    \end{align*}
    Notice that $\#_{S,n+1}(s_0\cdots s_{n-1})$ is the number of $n+1$-tuples that start with the $n$-tuple inside of $S$. However, this does not include the last $n$-tuple, $\vec x_l$, since there is no $n+1$-tuple in $S$ that starts with it. Similarly, if one does the same steps for $H_2'$, they find that they will be missing the first $n$-tuple, $\vec x_f$, since there is no $n+1$-tuple in $S$ that ends with that. Thus, rearranging and adding zero, we get:
    \begin{align*}
        H_2 &= \sum_{(s_0\cdots s_{n-1})\in \mathcal X^{n}} \#_{\hat S,n+1}(s_{n-1}\cdots s_0)\log(\hat p(s_{n-1},\cdots,s_0)) \\
        &\hspace{5em}+ \log(p(\vec x_f)) - \log(p(\vec x_l)) \\\\
        &= H'_2 + \log(p(\vec x_f)) - \log(p(\vec x_l))
    \end{align*}
\end{proof}
It doesn't appear possible to slightly adjust the quantity so that they are precisely equal. If one chose to remove the first observation for $S$ and $\hat S$ respectively, then the quantities would differ by a similar factor but based on $n+1$-tuples. The author has also numerically verified the differing constant on random subsets of Enwik9.

It should be possible to extend the main result to allow for continuous variables and time. These would permit study into more diverse data sources. It would also be interesting if there were some similar notion for non-sequential datasets.  

Since the choice of conditional $p$ was not special, we could have easily replaced it with two models $M$ and $\hat M$ trained on forward and backward passes. This is a rather surprising fact, and tells us that in theory compressing a file forwards and backwards should yield the same results \cite{shannon_mathematical_1948}. 

However, for the proof to work, one would have to assume that the training process is symmetric, in the sense that if one trains two models in the forwards and backwards directions, then the following equality will hold:
\begin{align*}
    M(s_n|s_0,\cdots,s_{n-1})p(s_0,\cdots,s_{n-1}) &= \hat M(s_0|s_n,\cdots,s_1)\hat p(s_{n},\cdots,s_1)
\end{align*}
We should expect this to be the case if one controls for the architecture and training methodology. However, if there is a large discrepancy, this would possibly imply that certain features are easier for the model to learn in one direction. Therefore, this gives rise not only for a way to control for the process that generates the dataset when testing learnability hypotheses, but also to quantify learnability itself. 

For practitioners, if one trains two identical models forwards and backwards on a large dataset, then computes the average cross entropy loss, the difference between the two quantities measures how much easier it is for one direction to learn:
\begin{align*}
    \Delta H = \frac{1}{N}(H_M(S)-H_{\hat M}(\hat S))
\end{align*}
If $\Delta H$ is positive, this implies the reverse direction is easier, and vice versa for negative. The absolute value measures how large of a difference this is and the usual interpretations of KL-divergence should transfer. 

If one has a hypothesis about which properties of a dataset make it easier for a model or training process to generalize or memorize \cite{zhang_understanding_2021}, then they could test this by constructing synthetic datasets that satisfy the property in only one direction. To remove noise, one may need to train $M$ and $\hat M$ multiple times and generate multiple datasets. A table of $\Delta H$ for various models and datasets may be useful to guide researchers.

Further, if $\Delta H$ is close to 0, but the above equality relating $M$ and $\hat M$ through $p$ fails badly, then that means the two models have learnt two different sets of features that perform at a similar level. Features that are different correspond exactly to the tuples which the equality breaks. Thus, this gives a way to determine whether the models are learning the same things. However, this is contingent on a good estimate of the unconditional probabilities being available. 
\section{Acknowledgements}
The author is grateful for Johns Hopkins' generous support during his time as an undergraduate. 

\printbibliography
\end{document}